\documentclass[aps,preprintnumbers,showpacs,superscriptaddress,pra]{revtex4-2}

\usepackage[margin=0.8in]{geometry}
\usepackage[utf8]{inputenc}
\usepackage[english]{babel}
\usepackage[T1]{fontenc}
\usepackage{amsmath,amssymb,amsthm}
\usepackage{mathtools}

\usepackage{graphicx}
\usepackage{url}

\usepackage{mathrsfs}
\usepackage{dsfont}
\usepackage{setspace}
\usepackage[export]{adjustbox}
\usepackage{makecell}
\usepackage{enumitem}

\usepackage{xcolor}
\definecolor{darkblue}{RGB}{0,0,128}
\definecolor{darkgreen}{RGB}{0,150,0}

\usepackage[pdfusetitle]{hyperref}
\hypersetup{breaklinks, colorlinks, linkcolor=blue, citecolor=darkgreen, filecolor=red, urlcolor=darkblue}
\usepackage{cleveref}

\newtheorem{theorem}{Theorem}
\newtheorem*{theorem*}{Theorem}
\newtheorem{lemma}[theorem]{Lemma}

\newtheorem{claim}[theorem]{Claim}
\newtheorem{corollary}[theorem]{Corollary}

\newtheorem{definition}[theorem]{Definition}

\usepackage{thmtools}
\usepackage{thm-restate}

\RequirePackage{filecontents}

\usepackage{mystyle}
\graphicspath{{./}{pics/}}

\usepackage{relsize}
\newcommand{\norm}[1]{\left| #1\right|}
\renewcommand{\epsilon}{\varepsilon}

\begin{document}
\title{Quantifying nonlocality: how outperforming local quantum codes is expensive}
\author{Nouédyn Baspin}
\affiliation{Universit\'e de Sherbrooke, Sherbrooke, Qu\'ebec, Canada J1K 2R1}
\author{Anirudh Krishna}
\affiliation{Stanford University, Stanford, CA, USA, 94305}

\date{\today}

  \begin{abstract}
  	Quantum low-density parity-check (LDPC) codes are a promising avenue to reduce the cost of constructing scalable quantum circuits.
    However, it is unclear how to implement these codes in practice.
    Seminal results of Bravyi \& Terhal, and Bravyi, Poulin \& Terhal have shown that quantum LDPC codes implemented through local interactions obey restrictions on their dimension $k$ and distance $d$. 
    Here we address the complementary question of how many long-range interactions are required to implement a quantum LDPC code with parameters $k$ and $d$.
    In particular, in 2D we show that a quantum LDPC code with distance $ d \propto n^{1/2 + \epsilon}$ requires $\Omega(n^{1/2 + \epsilon})$ interactions of length $\Omegalog(n^{\epsilon})$.
    Further a code satisfying $k \propto n$ with distance $d \propto n^\alpha$ requires $\Omegalog(n)$ interactions of length $\Omegalog(n^{\alpha/2})$.
    Our results are derived using bounds on quantum codes from graph metrics.
    As an application of these results, we consider a model called a stacked architecture, which has previously been considered as a potential way to implement quantum LDPC codes.
    In this model, although most interactions are local, a few of them are allowed to be very long.
    We prove that limited long-range connectivity implies quantitative bounds on the distance and code dimension.
  \end{abstract}
  \maketitle

\section{Introduction}
Finding ways to battle decoherence is among the foremost challenges on the path to implementing fault-tolerant quantum circuits.
Quantum error correcting codes can address this issue, and their efficacy is guaranteed by the quantum threshold theorem \cite{aharonov1997fault,kitaev1997quantum,knill1998resilient,aliferis2005quantum}.
The code we choose to use will be tailored to the advantages and disadvantages of the physical architecture on which it is implemented.
For instance, we might consider how many qubits we can measure jointly; how far apart qubits involved in such measurements need to be located; or how many supplementary qubits will be needed to implement a particular algorithm fault tolerantly \cite{ogorman2017quantum,sanders2020compilation}.
We will want the choice of code to be efficient and respect the limitations of our architecture.
Consequently, there is a strong interest in understanding how physical constraints on a system can impede the efficiency of a quantum code.

Formally, a quantum error correcting code $\cC$ on $n$ qubits is the common $+1$ eigenspace of a set of independent commuting $n$-qubit Pauli operators $\{\ssS_1,...,\ssS_m\}$, referred to as stabilizers, 
\begin{align*}
    \cC = \{\ket{\psi} : \ssS_i \ket{\psi} = \ket{\psi} \; \forall i \in \{1,...,m\} \}~.
\end{align*}
Measuring the stabilizers yields information required to detect and correct errors.
For ease of implementation, we may stipulate that these measurements be \emph{local} i.e.\ that the qubits involved in a stabilizer be contained within a ball of constant radius.
Let $k = \log_2 \dim \cC$ denote the number of encoded qubits \footnote{We often refer to $k$ as the number of \emph{logical} qubits, or as the \emph{dimension} of the code.}; we aim to encode as many qubits as possible with a limited number of available physical qubits.
Furthermore, let $d$ denote the distance; it is a measure of the number of physical qubits that need to be corrupted to irreparably damage encoded information.
Seminal works of Bravyi \& Terhal, and Bravyi, Poulin \& Terhal \cite{bravyi2009no,bravyi2010tradeoffs} demonstrated that there are sharp tradeoffs between $k$ and $d$ for all local codes.
As a result, locality limits our ability to reduce the resource cost of implementing scalable quantum circuits.
This naturally raises the following question---\textbf{Question 1:} to construct an error correcting code with dimension $k$ and distance $d$, how much nonlocality is needed to implement it?
How do we even quantify this seemingly nebulous notion of nonlocality?

Expanding our attention beyond local quantum codes is a worthwhile endeavor as certain architectures support interactions between arbitrary qubits.
Prominent examples are silicon-based architectures with photon-mediated interactions which encode qubits into the spin states of silicon \cite{bergeron2020silicon}, or photonic architectures where the qubits are directly encoded in the photons and therefore not localized \cite{bombin2021interleaving}.
Other architectures include atomic arrays \cite{periwal2021programmable}, where atoms are laid out along a single line, but long-range interactions can be used to simulate higher dimensions.
Ion trap architectures that support all-to-all connectivity in limited capacity has also been considered \cite{monroe2014large,linke2017experimental,murali2020architecting}.
By dropping the restriction of locality, these architectures can circumvent the limitations of local codes.
With this motivation, we consider quantum low-density parity-check (LDPC) codes, a class that subsumes all known topological codes \cite{kitaev1997quantum,bombin2006topological,bravyi1998quantum,kubica2015universal}.
The study of these codes is motivated by several results showing that quantum LDPC codes can drastically reduce the number of physical qubits required to build a fault-tolerant quantum computer \cite{gottesman2014fault,kovalev2013fault,fawzi2018constant}.
In practice, we wish to understand how to implement quantum LDPC codes in a $2$ or $3$-dimensional layout.
This then prompts the next question concerning locality---\textbf{Question 2:} can we implement good quantum LDPC codes using a setup where a majority of measurements are local?

In this paper, we address Questions 1 and 2.
Through Theorem \ref{thm:main} we show that quantum LDPC codes require large amounts of nonlocality between qubits when the dimension $k$ and the distance $d$ are large.
To motivate how to quantify nonlocality, we repeat an observation from \cite{baspin2021connectivity}.
It is not possible to add a limited number of long-range connections and significantly improve the performance of a local code.
Any code that we consider will have to have a sufficient number of long-range interactions to work.
Our quantification of nonlocality, therefore, in addition to the length of the long-range interactions, will also include the \emph{number} of such interactions.

We highlight codes for which $k \propto n$, and $d \propto n^\alpha$ for $\alpha > 0$, as these codes underpin the current proposals for low-overhead quantum computation.
Our results state that to implement these codes in 2D, we require roughly $n$ interactions of length $n^{\alpha/2}$. 
Therefore implementing these codes will require an architecture able to deal with a significant amount of nonlocality.
Although not known to exist, our results are also of interest for good codes i.e.\ constant-rate codes for which $\alpha = 1$.
If such codes exist, then they seem to make optimal use of long-range connectivity.
This is because in two dimensions the maximum distance between any two points on an $L \times L$ grid is proportional to $L \propto \sqrt{n}$, which would saturate our bound. 
Finally, our results suggest that it is expensive to outperform the distance of a local code.
For example, in 2D, Bravyi \& Terhal proved that local codes cannot do better than $d \propto n^{1/2}$; we show that any code satisfying $d \propto n^{1/2 + \epsilon}$ will require a growing number of long-range interactions.
Together, these results suggest that architectures limited to local interactions can only implement topological codes at best.

Next, as a model for implementations, we consider what we refer to as a \emph{stacked} layout \footnote{We are not aware of the origin of this model.
We heard about it through David Poulin;
something similar was also mentioned by Daniel Gottesman \cite{gottesman2021talk}.}.
This model is inspired by the schematic for a concatenated code shown in fig.~\ref{fig:stacked-pre}.
In the stacked model, qubits are placed on the vertices of a $2$-dimensional grid.
The measurements required to define the code are partitioned into multiple layers as visualized in fig.~\ref{fig:stacked-pre}.
Each layer of the stack represents stabilizers of a given interaction radius.
\begin{figure}[h]
	\centering
	\includegraphics[scale=0.12]{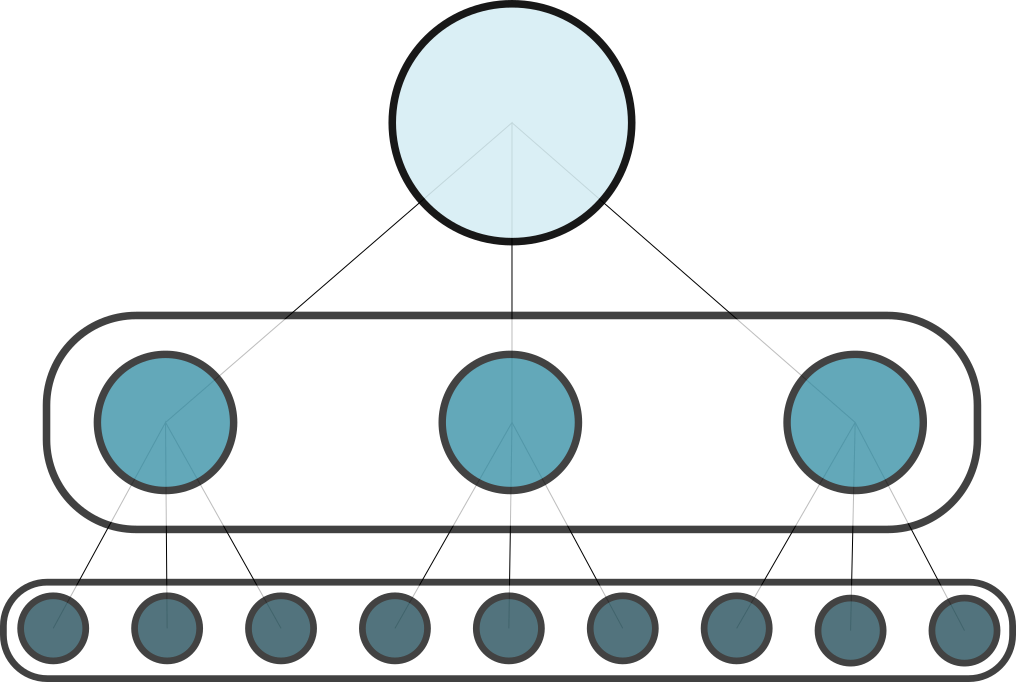}
    \includegraphics[scale=0.135]{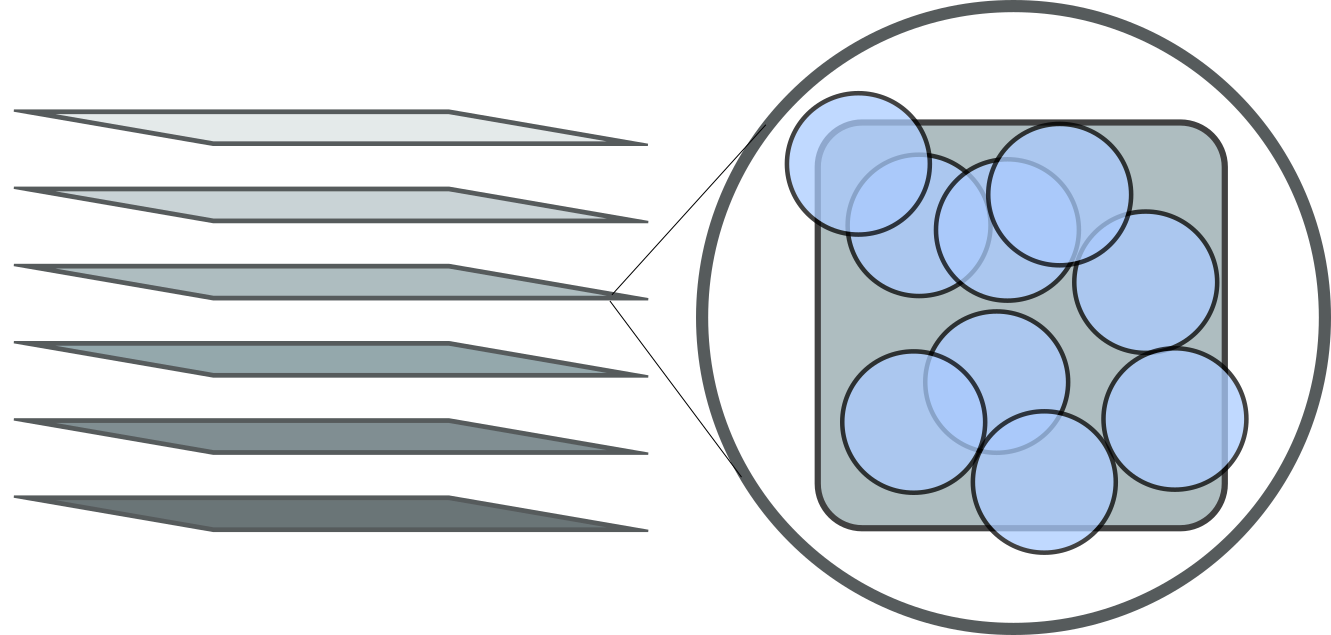}
	\caption{(a) A schematic for a concatenated code.
    The qubits of the code are themselves encoded in an error correcting code and this gives rise to a hierarchical structure.
    (b) A $2$-dimensional stacked architecture.
    Qubits are the bottom-most layer.
    Stabilizers, identified with their support, are assigned to different layers above and are depicted using blue circles.
    Stabilizers in a given layer have a radius of support depending on the layer.
    This interaction range increases as we move up the stack or equivalently the radius of the circles increases.
    On the other hand, the number of stabilizers in each layer decreases exponentially.}
	\label{fig:stacked-pre}
\end{figure}
The interaction range increases as we move up the layers of the stack while the number of stabilizers decreases.
The majority of stabilizers in this model are in the lower layers.
Therefore any code implemented by a stack is mostly local.
For this reason, this model has been considered a potential route to implement LDPC codes.
However, such an architecture cannot implement arbitrary quantum LDPC codes.
We show that $2$-dimensional stacked layouts are limited.
The distance is bounded by $d = \widetilde{O}(n^{2/3})$ and the dimension-distance tradeoff is $k^{3}d^{4} = \widetilde{O}(n^{5})$.
This is presented in Section \ref{sec:stacked}.
This shows that there are strong limitations to such models; however, it does not prevent implementations of constant-rate codes with distance scaling as $\sqrt{n}$.
Therefore, it may be possible to implement hypergraph product codes \cite{tillich2014quantum}.

\section{Background and intuition}

An $\dsl n, k, d \dsr$ quantum code $\cC$ is a $2^k$-dimensional subspace of the complex Euclidean space $\bbC^{2^n}$ associated with $n$ qubits.
The distance $d$ is the minimum number of qubits that are acted on nontrivially by a unitary operation to map one element of $\cC$ to another.
The codespace is specified as the joint $+1$-eigenspace of a set of commuting Pauli operators $\cS \subset \{\ssI,\ssX, \ssY, \ssZ\}^{\otimes n}$ called the stabilizer group.
Suppose the group is generated by some elements $\{\ssS_i\}_{i=1}^{n-k}$.
The code is said to be a low-density parity-check (LDPC) code if each generator only acts on a constant number of qubits, and each qubit is only involved in a constant number of generators.

We represent a quantum code $\cC$  on $n$ qubits using a \emph{connectivity graph} $G = G(\cC) = (V,E)$.
Each vertex $v \in V$ of $G$ corresponds to a qubit of $\cC$ and two vertices share an edge $e \in E$ if both qubits participate in the same measurement $\ssS_i$.
We quantify the connectivity of $G$ using the notion of a graph separator.
A separator $\sep(G) \subseteq V$ is a subset of vertices which, if removed, would split $G$ into two subgraphs that are disconnected from each other.
In other words, we may write the vertices of $G$ as a disjoint union $V = A \sqcup \sep(G) \sqcup B$ such that there are no edges between $A$ and $B$.
Furthermore, we require that $|A|, |B| \leq 2n/3 = 2 |V|/3$.

We define the separation profile $s_G : \mathbb{N} \rightarrow \mathbb{N}$ of a graph $G$ to be $s_G(r) = \max_{H \subseteq G, |H| \leq r} |\sep(H)|$. When we consider a family $\cG = \{G_n\}_n$, we write $s_n$ the separation profile of $G_n$. In the following sections, we will assume that $s_n(r) > 0$, or equivalently that the code is not trivial.
Of particular interest are families of expander graphs which we will return to later.
These graphs are very well connected; the separator for a family of expander graphs scales in proportion to the size of the graph, i.e.\ $s_n \propto n$.

In \cite{baspin2021connectivity}, we showed that there is an intimate relationship between the properties of a quantum code and the corresponding connectivity graph.
The connectivity here is quantified by the size of the separator.
Our result, stated formally below in Lemma \ref{lem:generalizedbounds}, has two parts to it.
First, the size of the separator bounds the distance of the code.
Next, the smaller the separator, the sharper the tradeoff between code parameters $k$ and $d$.
In its general form, the bound is stated in terms of the following quantities. 

\begin{definition}
    \label{def:cminmax}
    Let $\scrC = \{\cC_n\}_n$ be a family of $\dsl n,k(n),d(n)\dsr$ quantum LDPC codes with nontrivial connectivity graphs $\cG = \{G_n\}_n$ with associated separation profiles $\{s_n\}_n$.
    Consider the quantity $c_n(r) \equiv \log_r(s_n(r))$.
    For each $G_n \in \cG$, define the quantities $c_{\max}(n)$ and $r_0(n)$ as
    \begin{align*}
        c_{\max}(n) =  \max_{r \in [d,n]} c_n(r) \qquad r_0(n) = \arg\max_{r \in [d,n]} c_n(r)~.
    \end{align*}
\end{definition}
The quantity $c_{\max}$ measures how tightly the graph is connected by considering subgraphs of whose size lies in the interval $[d,n]$.
Consider fig.~\ref{fig:cmax} which shows a connectivity graph $G$ and subgraphs $H_0 \subseteq H_1 \subseteq G$.
Any subgraph $H_0$ such that $|H_0| = d$ may itself be tightly connected.
However, as we increase the size of the subgraph to $H_1$ and to $G$, the graphs may be easy to separate.
\begin{figure}[h]
    \centering
    \includegraphics[scale=.5]{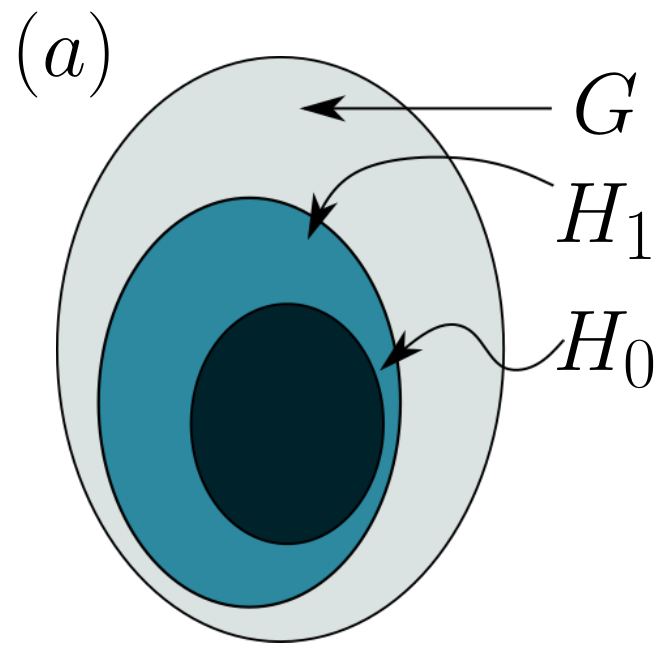}
    \hspace{15mm}
    \includegraphics[scale=.5]{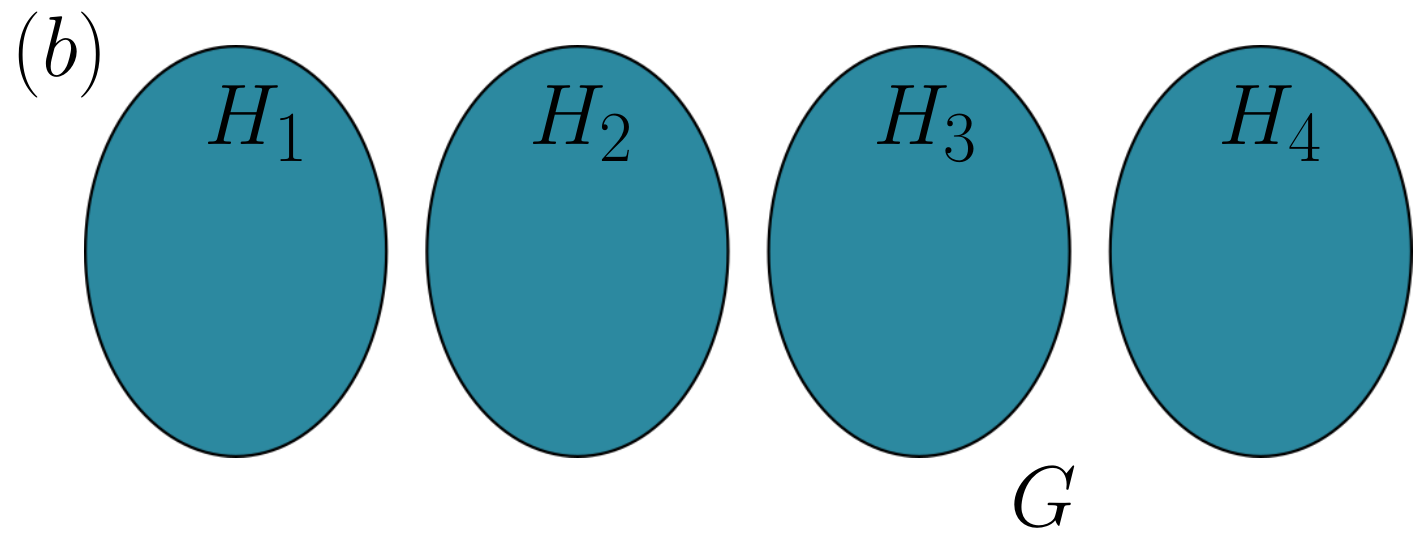}
    \caption{Visualizing $c_{\max}$ for a connectivity graph $G$.
    In (a) the subgraphs $H_0 \subseteq H_1 \subseteq G$ have sizes $d \leq |H_0| \leq |H_1| \leq |G|$.
    $H_0$ itself may be very tightly connected, but $H_1$ and $G$ need not be.
    As an example, consider (b) where the connectivity graph $G$ is made up of several disconnected expander subgraphs of size $n^{\alpha}$; there are $n^{1-\alpha}$ such subgraphs, and therefore $G$ has $n$ vertices in total.
    This implies that $c_{\max} = 1$.}
    \label{fig:cmax}
\end{figure}
As an example, consider a connectivity graph $G$ which corresponds to a set of $n^{1-\alpha}$ disconnected expander graphs as shown in fig.~\ref{fig:cmax} (b).
Each subgraph has size $n^{\alpha}$, and so $G$ has $n$ vertices in total.
Suppose it was known that $d = O(n^{\alpha})$.
If we consider small enough subgraphs, i.e.\ of size lesser than $n^{\alpha}$, then there exist subgraphs with large separators.
However, if we let $r = n^{\alpha}$, the largest separator corresponds to any subgraph $H_i$ and therefore $c_{\max} = 1$.
This sort of connectivity graph might show up naturally for example in Gottesman's construction of a fault-tolerant quantum circuit using LDPC codes \cite{gottesman2014fault}.
We say more about this construction following Theorem \ref{thm:main}.

For the sake of readability, we will simply write $k \equiv k(n)$ and $d \equiv d(n)$.
Note that there exists a subgraph $H \subseteq G_n$ such that $d \leq |H| \leq n$ and $|\sep(H)| = |H|^{c_{\max}}$.
We also note that for all $r$, we have that $s_n(r) \leq r^{c_{\max}(n)}$.
Further one can note that $|\sep(G)| \leq |V|/3$ for any graph $G$, as any set $S$ of size $|V|/3$ always induces a set $A = \emptyset$, and $B = V \setminus S$, such that $|A|, |B| \leq 2|V|/3$. Therefore we always have $c_{\max}(n) \leq \log_n(1/3) + 1 < 1$.

These quantities allow us to express the bounds on codes given the connectivity graph representation as presented in \cite{baspin2021connectivity}.
\begin{restatable}[Generalized bounds on codes]{lem}{generalizedbounds}
	\label{lem:generalizedbounds}
    Let $\scrC = \{\cC_n\}$ be a family of $\dsl n,k,d\dsr$ quantum LDPC codes with nontrivial connectivity graphs $\cG = \{G_n\}_n$.
    Let $c_{\max}(n)$ and $r_0(n)$ be defined as above, then 
	\begin{align*}
		d = O(s_n(n)) ~, \qquad k = O(d^{2(c_{\max}(n)-1)}n \log(n)^2)~.
    \end{align*}

	Further, if we have $c_{\max}(n) \leq c_0$ for a constant $c_0 \in (0,1)$, then
	\begin{align*}
		k = O(d^{2(c_{\max}(n)-1)}n)~.
	\end{align*}
\end{restatable}
For an in-depth discussion of this lemma, including the proof, we point the interested reader to \cite{baspin2021connectivity}.

\section[]{Embedding codes in \texorpdfstring{$D$}{Lg}-dimensions}

In this section, we consider how to embed quantum LDPC codes in $\bbR^D$.
This section is inspired by results from metric geometry that consider the distortion of expander graphs embedded in $\bbR^D$.
Here we show that a class of graphs called $\epsilon$-expanders are difficult to embed.
As a consequence, we show that constant-rate quantum codes require a growing number of long-range interactions between qubits.

\begin{definition}
	For a graph $G = (V,E)$, a map $\eta: V \rightarrow \bbR^D$ is called an embedding.
    Further, $\eta_{\theta}$ is a $\theta$-embedding if it satisfies the following condition for all pairs of distinct vertices $u,v \in V$,
	\begin{align*}
	    \norm{\eta_{\theta}(u) - \eta_{\theta}(v)} \geq \theta~.
	\end{align*}
    We use $|\, .\, | : \bbR^D \rightarrow \bbR$ to denote the standard Euclidean metric.
\end{definition}
For example, if we were to embed qubits in a $2$-dimensional grid with the points unit distance apart, we would let $\theta = 1$.
In the following sections, we will frequently refer to the length of an edge.
We mean that any embedding $\eta_\theta$ naturally endows an edge $(u,v)$ with a length.
Equivalently, the length of an edge $(u,v)$ is $|\eta_\theta(u) - \eta_\theta(v)|$. 

\begin{theorem}[Main]
    \label{thm:main}
    Let $\scrC = \{\cC_n\}$ be a family of $\dsl n,k, d \dsr$ quantum LDPC codes
    Further suppose $\scrC$ is associated with the nontrivial connectivity graphs $\cG = \{G_n = (V_n, E_n)\}_n$.
    For any $\theta$-embedding $\eta_{\theta}: V_n \to \bbR^D$, there exists some $\beta, n_0$ such that for code sizes $n > n_0$, and any $\alpha \in (0,1)$, the following propositions hold:
	\begin{enumerate}
		\item $\eta_{\theta}$ induces $\Omega(d)$ edges of length $\Omegalog\left(\frac{d}{n^{(D-1)/D}}\right)$.
		\item $\eta_{\theta}$ induces $\Omegalog\left(\sqrt{\frac{k}{n}}d\right)$ edges of length $\Omegalog\left(\sqrt{\frac{k}{n}}d^{1/D}\right)$.
		\item If $kd^{2/D} \geq \beta  n \log(n)^2/(1-\alpha)$, then $\eta_{\theta}$ induces at least $\Omegalog\left(\sqrt{\frac{(1-\alpha)k}{n}}^{1/\log_n(d)} \alpha k \right)$ edges of length $\Omegalog\left(\sqrt{\frac{(1-\alpha)k}{n}}d^{1/D}\right)$.
	\end{enumerate}
\end{theorem}

The proof is presented is Section \ref{subsec:proofmain} and Sections \ref{subsec:epsexp} and \ref{subsec:tdens} establish the tools required.
We first present a short discussion of the theorem.

\textbf{Discussion:}

As a reminder, an edge of length $l$ implies that there exist a stabilizer measurement involving at least two qubits which are embedded at a distance at least $l$ from each other.
We say that such stabilizer has range at least $l$.
If an embedding induces $m$ edges of length $l$, then, since the codes we consider are LDPC, there exist at least $\Theta(m)$ stabilizers of range at least $l$.

\begin{enumerate}
	
\item We focus on the case $D = 2$.
The first observation is that a code of distance $\Omega(n^{1/2 + \epsilon})$ will induce $\Omega(n^{1/2 + \epsilon})$ edges of length $\Omegalog(n^{\epsilon})$ from Claim 1.
This underlines how hard it is to break free of the natural restrictions space imposes on the distance: the case $\epsilon = 0$  can be obtained readily using topological codes and only nearest neighbors interactions, but $\epsilon > 0$ will require a significant amount of nonlocality.
In particular, implementing a linear distance code will induce $\Omega(n)$ edges of length $\Omegalog(n^{1/2})$.
In that particular case, the length of the edges are tight up to logarithmic factors, since any code can be implemented on a $\sqrt{n} \times \sqrt{n}$ square lattice such that all qubits are at a distance at most $O(n^{1/2})$ from each other.
In $D$ dimensions, this result can also be seen as a more general version of the Bravyi-Terhal claim \cite{bravyi2009no}---if the code is local, then the longest edges of its connectivity graph have length $O(1)$, and applying Claim 1 implies that $d = \Olog(n^{(D-1)/D})$.

\item Similarly, our results yield nontrivial bounds on codes with constant rate.
First, consider the case with $k \propto n$ and $d \propto 1$.
Such a code can be achieved using $\Theta(n)$ disjoint patches of a $2$D topological code, and this implementation requires zero nonlocal interactions.
However, Claim 3 shows that escaping from this constant distance is challenging.
For example, achieving $d \propto n^\alpha$ requires $\Omegalog(n)$ interactions of length $\Omegalog(n^{\alpha/2})$: quite a dramatic change. Similarly as in the previous point, Claim 2 \& 3 can be read as a weaker and more general version of BPT: if the size of the interactions are in $O(1)$, then $kd = \Olog(n)$.

\item If good codes exist --- codes for which $k,d \propto n$ --- then they seem to make optimal use of nonlocality, as they almost saturate Claim 3.
For example, we could implement $n^{1-\alpha}$ disjoint blocks of good codes, each with size $n^{\alpha}$.
Then we have $k \propto n$, $d \propto n^{\alpha}$, and at most $O(n)$ edges of length $n^{\alpha/2}$, which minimizes the bound as discussed in the previous point.
This suggests that if good quantum codes were to exist, they will likely be essential in decreasing the experimental cost of quantum error correction.

\item As previously mentioned, there is a notable gap between BPT and our results: $kd^2 = O(n)$ in the first case, and $kd = \Olog(n)$ in the second. 
It is then worth asking if we can close the gap between these bound.
Can Claim 2 \& 3 be sharpened to yield a nontrivial bound on codes satisfying $kd^{2/(D-1)} = \Omega(n)$?
This question does not seem to be trivial to us.
Suppose we naively substitute $\sqrt{k/n}d^{1/D}$ by $\sqrt{k/n}d^{1/(D-1)}$ such that $O(1)$ interactions imply $kd^2 = O(n)$.
Then, in 2D, for any distance larger than $n^{1/2+\epsilon}$ and constant rate, we get some edges that are larger than $n^{1/2+\epsilon}$.
However, this is impossible: we can always place the qubits in a $\sqrt{n} \times \sqrt{n}$ square with edges of length $O(\sqrt{n})$.
This seems to imply that if that substitution worked, there exists no constant-rate quantum LDPC code with a distance larger than $\sqrt{n}$, which would be surprising.

\item We conjecture that the $\log(n)$ factors in the length of the edges is suboptimal.
In other words, we believe that $\Omegalog$ is just $\Omega$.
We are, however, unable to prove this with our techniques.

\item 
What do these results mean for fault-tolerant quantum computation?
To be concrete, consider Gottesman's construction \cite{gottesman2014fault}.
It allows us to use any LDPC code of rate $R$; we partition the $k$ logical qubits we wish to process into several blocks, each of which is encoded with an LDPC code.
Each block has size proportional to $k/(R\;\poly\log(k))$.
If we were to use a hypergraph product code class in this construction, then such codes can achieve a constant rate and a distance that scales as the square root of the total number of qubits.
How difficult would it be to embed this construction in two dimensions?
It would require a constant fraction of vertices that are connected to edges of length $\Omega((k/R\;\poly\log(k))^{1/4})$.
In this sense, Gottesman's construction does not require edges as long as would be needed if computation were performed on a single block of LDPC code
\footnote{Our understanding of how to perform fault-tolerant quantum computation using only a single block of an LDPC code is limited; see for example \cite{krishna2021fault}.}.

\end{enumerate}

\subsection[]{\texorpdfstring{$\epsilon$}{Lg}-expansion and \texorpdfstring{$D$}{Lg}-dimensional embeddings}
\label{subsec:epsexp}
In the previous section, we considered graph families $\cG$ with separation profiles $\{s_n\}_n$.
However, separability is not directly amenable to discussions of embeddings.
We find it useful to work with a slightly different notion of graph connectivity called $\epsilon$-expansion.
We show that if a graph is an $\epsilon$-expander, then every embedding $\eta_{\theta}$ that embeds the graph in $D$ dimensions will induce roughly $n\epsilon$ edges of length roughly $n^{1/D}\epsilon$ (up to log factors).
To make the connections explicit, we begin by restricting our attention to individual graphs.
We then extend these results to graph families.

As intuition for this subsection, we note that we cannot arbitrarily embed any graph on a grid and always expect the embedding to preserve the lengths of edges.
The extent to which these lengths can change is captured by what we define as the stretch.
If the stretch is large, this implies that there exists an edge in the embedding spanning a large distance.
As the vertices in the connectivity graph represent qubits, a long edge means that two qubits that are far from each other are involved in the same measurement.

\begin{definition}
	For an embedding $\eta_{\theta}: V \to \bbR^D$, we define the stretch $\dist(\eta_{\theta})$ as the longest induced distance between any two neighboring vertices. 
	\begin{align*}
	    \dist(\eta_{\theta}) = \max_{(u,v)\in E} \norm{\eta_{\theta}(u)- \eta_{\theta}(v)}~.
	\end{align*}
\end{definition}
The stretch is just the longest length of an edge $(u,v) \in E$ under the embedding $\eta_{\theta}$.
It can be used to upper bound the distance between two qubits in the physical space $\bbR^D$.
This quantity is inspired by and related to the distortion of an embedding \cite{matouvsek2013lecture}.
However, to avoid confusion, we have given this quantity a different name.
The first difference is that we are evaluating this quantity only over the edges of the graph.
Second, we disregard the amount by which the distances between points may be contracted.
This is all that matters for our purposes.
The distortion itself is the product of what we call the stretch evaluated for any two vertices and what might be called the contraction: the extent to which the distance between points has been shrunk.

\begin{definition}[$\epsilon$-vertex-expansion]
    \label{def:epsexpansion}
    Consider a graph $G = (V,E)$ on $n$ vertices.
    For $A \subseteq V$, let $\bdry A$ be the number of vertices of $V \setminus A$ that are connected to $A$, i.e.\
    \begin{align*}
        \bdry A = \{v \in V\setminus A : \exists u \in A, (u,v) \in E\}~. 
    \end{align*}
    We say $G$ is an $\epsilon$-expander, $\epsilon \in [0,1]$, if  
	\begin{align*}
	\min_{A \subset V, |A| \leq |V|/2} \frac{|\bdry A|}{|A|} \geq \epsilon
    \end{align*}
\end{definition}
Observe that $\epsilon \leq 1$ when $n$ is even, because we can let $|A| = |V|/2$, and $\epsilon \leq 1 + 2/(n-1)$ when $n$ is odd.

The proof that $\epsilon$-expanders are difficult to embed will involve a packing argument.
We first quantify how large the embedding of a graph will be under $\eta_{\theta}$.
The stretch implies an upper bound on how far two arbitrary vertices can be.
This is captured by the following lemma.

\begin{lemma}
	\label{lem:stretch-bounds-distance}
	Let $G = (V,E)$ and $\eta_{\theta}: V \to \bbR^{D}$ be a $\theta$-embedding, then for all $u,v \in V$
	\begin{align*}
        \norm{\eta_{\theta}(u) - \eta_{\theta}(v)} \leq \dist(\eta_{\theta}) d_G(u,v)~.
	\end{align*}
\end{lemma}
\begin{proof}
	For any two $u,v \in V$, let $P(u,v) = \{u_0,u_1,u_2,...,u_{d_G(u,v)}\}$ be the path of minimum length between $u =: u_0$ and $v =: u_{d_G(u,v)}$.
	By the triangle inequality, it follows that
	\begin{align*}
		\norm{\eta_{\theta}(u) - \eta_{\theta}(v)} \leq \sum_{i=0}^{d_G(u,v)-1}\norm{\eta_{\theta}(u_i) - \eta_{\theta}(u_{i+1})} \leq \dist(\eta_{\theta}) d_G(u,v)~.
	\end{align*}
    This concludes the proof.
\end{proof}

The maximum distance between any two vertices is called the diameter $\diam(G)$ of a graph $G$,
\begin{align*}
    \diam(G) = \max_{(u,v)\in E} d_G(u,v)~.
\end{align*}
\begin{claim}
    \label{claim:epsilondiam}
    Let $G$ be an $\epsilon$-expander of bounded degree.
    Then $\diam(G) = O\left(\log(n)/\epsilon\right)$.
\end{claim}
\begin{proof}
	From the Handbook of Linear Algebra \cite{hogben2013handbook},
	\[
	    \diam(G) = O\left(\sqrt{\frac{\delta}{\lambda_2}}\log(n)\right)~.
	\]
	With $\delta$ the maximum degree of the graph $G$, and $\lambda_2$ the second smallest eigenvalue of the Laplacian $\cL(G)$.
    It can then be read from the Cheeger inequality \cite{alon2016probabilistic} that $\epsilon \leq \sqrt{2\lambda_2}$.
    Hence the desired result.
\end{proof}

Our next result states that a $\theta$-embedding of an $\epsilon$-expander will necessarily have at least one long edge.
The following packing argument shows that if we want to pack balls of radius $\theta$ into a $D$-dimensional space, the entire graph will require a certain volume.
In turn, this implies a lower bound on the stretch.

\begin{lemma}
	\label{lem:euclidstretch}
	Let $G=(V,E)$ be an $\epsilon$-expander, and $\eta_{\theta}:V \rightarrow \bbR^D$ a $\theta$-embedding for $\theta > 0$.
    Then $\dist(\eta_{\theta}) = \Omega\left(\frac{n^{1/D}\epsilon}{\log(n)}\right)$.
\end{lemma}
\begin{proof}
    Let $\cB(x,r) \subset \bbR^D$ denote the ball with center $x$ and radius $r$.
    
    As we have just seen from Lemma \ref{lem:stretch-bounds-distance}, we can assert that for all $u,v \in V$, that $\norm{\eta_{\theta}(u) - \eta_{\theta}(v)} \leq \dist(\eta_{\theta}) \cdot \diam(G)$.
    Equivalently, for some $u \in V$ and any other $v \neq u \in V$, $\eta_{\theta}(v) \in \cB(\eta_{\theta}(u), \dist(\eta_{\theta}) \diam(G))$: all the $n$ vertices are mapped to points contained in this ball.
    
    On the other hand, the definition of $\eta_{\theta}$ implies that every $u \in V$ induces an empty ball $\cB(\eta_{\theta}(u), \theta)$ around $u$, i.e.\ for any $v \in V$, $u\neq v$, $\eta_{\theta}(v)\not\in \cB(\eta_{\theta}(u), \theta)$.
	
	This implies that the number of balls $\cB(\eta_{\theta}(v), \theta/2)$ contained in $\cB(\eta_{\theta}(u), \dist(\eta_{\theta}) \diam(G))$ has to be at least equal to $n$.	
	Letting $\vol{\cB(x,r)}$ denote the volume of the ball $\cB(x,r)$ in $\bbR^D$, we have
	\begin{align*}
		\vol{\cB(\eta_{\theta}(u),  \dist(\eta_{\theta}) \diam(G))} \geq n \cdot \vol{\cB\left(0, \frac{\theta}{2} \right)}~.
    \end{align*}
	For fixed $D$, we have $\vol{\cB(x, r)} = c_D r^D$, for a constant $c_D$ depending only on $D$.
    We also define $\vol{\cB(0, \theta/2)} \equiv c_\theta$, with $c_\theta$ depending only on $\theta$ and $D$.
	Substituting this above, we get
	\begin{align*}
        c_D (\diam(G) \dist(\eta_{\theta}))^D &\geq c_\theta n \\
	\implies \dist(\eta_{\theta}) &\geq \alpha_{\theta, D}\frac{n^{1/D}}{\diam(G)}
    \end{align*}
    for some constant $\alpha_{\theta,D}$ as desired.
    Finally, we note from Claim \ref{claim:epsilondiam} that if $G$ is an $\epsilon$-expander, then $\diam(G) = O(\log(n)/\epsilon)$.
    This concludes the proof.
\end{proof}

This result proves that \emph{some} qubits will require long-range connections.
To prove Theorem \ref{thm:main}, we will show that expansion is robust to the removal of a large number of edges. 

\begin{lemma}
	\label{lem:sepofexp}
	For any $\epsilon$-expander graph $G$ on $n$ vertices, $|\sep(G)| \geq \min(\frac{1}{6}n, \frac{1}{6}n\epsilon)$. 
\end{lemma}
\begin{proof}
	Assume a separator $\sep(G)$ inducing a partition $A \sqcup \sep(G) \sqcup B$, such that $|\sep(G)| \leq n/6$.
    If not, the bound is already true.
    Without loss of generality, suppose $|A| \leq |B|$.
    By definition, we have $|B| \leq 2n/3$, and thus $|A| + |\sep(G)| = |V| - |B| \geq n/3$.
    From the upper bound on $|\sep(G)|$, we have $|A| \geq n/6$.
    As $\bdry A \subset \sep(G)$, we have $|\sep(G)| \geq |A|\epsilon \geq \frac{1}{6}n\epsilon$.
\end{proof}

The traditional expander graph that corresponds to $\epsilon$ being constant will always have the worst stretch.
In other words, these are the hardest graphs to embed in $D$-dimensions.
Therefore, for the sake of readability, we will hereafter assume that $\epsilon \leq 1$, and $|\sep(G)| \geq \frac{1}{6}n\epsilon$. 

It is possible to find a converse to Lemma \ref{lem:sepofexp}.

\begin{lemma}[Lemma 12 from \cite{bottcher2010bandwidth}]
	\label{lem:bottsep}
	Let $G$ be a graph on $n$ vertices and $\phi \in [0,1]$.
    If $|\sep(G)| \geq \phi n$, then there exists a subgraph $G' \subset G$ such that $|G'| \geq \frac{n}{2}$, and $G'$ is a $\frac{2\phi}{3}$-expander.
\end{lemma}

\begin{lemma}
	\label{lem:stillexp}
	Let $G$ be an $\epsilon$-expander. 
	It is possible to arbitrarily remove $\frac{1}{24}n\epsilon$ edges such that the remaining graph $G'$ still contains a subgraph that is $\frac{1}{18}\epsilon$-expander and of size $\Theta(n)$.
\end{lemma}
\begin{proof}
    We start by noting that an $\epsilon$-expander graph $G$ has no separator smaller than $\frac{1}{6}n\epsilon$ from Lemma \ref{lem:sepofexp}.
    
    Let us then pick a constant $\alpha \in [0,1)$, and we remove a subset $F \subset E$ of edges from $G$, such that $|F| =  \frac{\alpha}{2} \frac{1}{6}n\epsilon$. 
    Then the subgraph $G' = (V,E \setminus F)$ we obtain has no separator smaller than $(1-\alpha)\frac{1}{6}n\epsilon$. 
    
    To convince ourselves of this proposition, let's assume the contrary, and let $S' \subset V$ be a separator of $G'$ such that $|S'| < (1-\alpha)\frac{1}{6}n\epsilon$. 
    The separator $S'$ can then be promoted to a separator of $G$ in the following manner. Consider the set $S_{F}$, the set of endpoints of the edges in $F$, then $S_{F} \cup S'$ is a separator of $G$. Indeed, one can go from $G'$ to $G$ by adding back the edges $F$. Therefore if $S'$ yields two disconnected partitions of $G'$, then so does $S_{F} \cup S'$.
    
    We then have $S_{F} \cup S'$ a separator of $G$, and $|S_{F} \cup S'| \leq |S_{F}| + |S'| \leq 2|E'| + |S'| < \frac{1}{6}n\epsilon$, which contradicts our assumption that $G$ has no separator smaller than $\frac{1}{6}n\epsilon$. Therefore $G'$ has no separator smaller than $(1-\alpha)\frac{1}{6}n\epsilon$.
    
    We can then use the fact that a graph with no small separator contains a large expander within. Since $G'$ has no separator smaller than $(1-\alpha)\frac{1}{6}n\epsilon$, then $G'$ has a subgraph $G''$ of size $n/2$, such that $G''$ is $(1-\alpha)\frac{1}{9}\epsilon$-expander from Lemma \ref{lem:bottsep}. 
    
    By fixing $\alpha = \frac{1}{2}$ we obtain the desired result.
\end{proof}

We can obtain a stronger version of Lemma \ref{lem:euclidstretch}: an expander graph has a large number of long edges for any embedding in a low dimension.

\begin{lemma}
	\label{lem:noncontractedges}
	Let $G$ be $\epsilon$-expander, and $\eta_{\theta}: V \rightarrow \bbR^D$ a $\theta$-embedding.
    Then $\eta_{\theta}$ induces $\frac{1}{24}n\epsilon$ edges of length $\Omega(n^{1/D}\epsilon/\log(n))$.
\end{lemma}

\begin{proof}
	We embed $G$ in $\bbR^D$ with $\eta$. We can then order the edges of $G$ by their length in decreasing order $\{e_1,...,e_{|E|}\}$.
    By Lemma \ref{lem:stillexp}, we can remove the first $\frac{1}{24}n\epsilon$ of these edges and the resulting graph still contains a subgraph $G''$ such that $|G''| = \Theta(n)$, and $G''$ is $\frac{1}{18}\epsilon$-expander. 
	And therefore, from Lemma \ref{lem:euclidstretch}, this subgraph still contains an edge $e$ of size $\Omega(n^{1/D}\epsilon/\log(n))$.
    From our assumptions, the $\frac{1}{24}n\epsilon$ edges we removed are all longer than $e$, and the desired result follows.
\end{proof}

\subsection[]{\texorpdfstring{$t$}{Lg}-density}
\label{subsec:tdens}

We have seen in the previous section that separation and expansion are very closely related. 
We thus expect a result similar to Lemma \ref{lem:noncontractedges} to hold for graphs with a large separation profile. 
A connection between separation profiles and $\theta$-embeddings would then be able to tell us something about the embeddings of quantum codes: by Lemma \ref{lem:generalizedbounds}, a quantum code with good parameters needs its separation profile to be large. 
In order to lighten the notation, we will use the notion of $t$-dense graphs to describe large separation profiles.
\begin{definition}
	\label{def:tdense}
	A graph $G$ on $n$ vertices is said to be $t$-dense, $t \in \mathbb{N}$, if its separation profile $s_G$ satisfies $s_G(n) \geq t$.
\end{definition}

This definition captures the notion of a graph that has some well-connected subgraph.
It can readily be seen that if a graph is $t$-dense, then it contains a subgraph $H \subseteq G$ such that $H$ has no separator smaller than $t$, or equivalently, $|\sep(H)| \geq t$. 
Otherwise, by definition its separation profile would satisfy $s(n) \leq t-1$, which contradicts Definition \ref{def:tdense}.

We can then formalize the connection between separation profiles and expansion.

\begin{lemma}
	\label{lem:tdenseexpander}
	Let $G$ be a graph on $n$ vertices such that $G$ is $t$-dense.
    Then there exists a subgraph $H' \subset G$ such that $|H'| \geq t/2$, and $H'$ is a $\frac{t}{3 |H'|}$-expander.
\end{lemma}
\begin{proof}
	As previously mentioned, there must exist a subgraph $H \subset G$ such that $|\sep(H)| \geq t$.
	We can then pick $\epsilon$ such that $\epsilon |H| = t$, and verify that since $|\sep(H)| \geq t$, then $|\sep(H)| \geq \epsilon |H|$.
    From Lemma \ref{lem:bottsep}, there exists a subgraph $H' \subset H$, such that $|H'| \geq \frac{|H|}{2}$, and $H'$ is $\frac{2}{3}\epsilon$-expander. 
	Since $H'$ is $\frac{2}{3}\epsilon$-expander, i.e. $\frac{2t}{3|H|}$-expander, then it is also $\frac{t}{3|H'|}$-expander, as $\frac{1}{|H|} \geq  \frac{1}{2 |H'|} $.
	This holds due to the following observation: if $H'$ is an $\epsilon$-expander, and $\epsilon \geq \epsilon' $, then it is also $\epsilon'$-expander.
	
	Finally note that since $|\sep(H)| \geq t$ then $|H| \geq t$, which gives $|H'| \geq t/2$.
\end{proof}

Since a $t$-dense graph induces a large expander subgraph, we expect this expander subgraph to induce a large number of long edges.

\begin{lemma}
	\label{lem:denseedges}
	Let $G$ be a graph on $n$ vertices such that $G$ is $t$-dense for some $t \in \bbN$.
    Then the $\theta$-embedding $\eta_{\theta}$ induces $\Omega(t)$ edges of length $\Omega\left(\frac{t}{\log(n)n^{1-1/D}}\right)$.
\end{lemma}
\begin{proof}
	From Lemma \ref{lem:tdenseexpander}, we know that there exists $H' \subset G$, such that $|H'| \geq t/2$, and $H'$ is a $\frac{t}{3|H'|}$-expander. 
	
	Further, from Lemma \ref{lem:noncontractedges}, in order to embed $H'$, we need at least $\frac{1}{24}t$ edges of length $\Omega\left(\frac{t}{|H'|^{1-1/D}\log(n)}\right)$. 
	
	The length of the edges increases as $|H'|$ decreases.
    It is always true that $|H'| \leq n$ which yields the desired result.
\end{proof}

Interestingly, the smaller $|H'|$ is, the longer the edges are.
For example, imagine a graph that is $\sqrt{n}$-dense.
If $|H'| \propto \sqrt{n}$ then $H'$ is like an expander of size $\sqrt{n}$.
On the other hand, if $|H'| = n$, then $H'$ might just be planar.
In this way, the smaller graph of a fixed density, i.e.\ the expander, is harder to embed.

The contrapositive of Lemma \ref{lem:denseedges} can be interpreted as a separator theorem for local graphs of bounded degree.

Finally, it only remains to show that a quantum code with good parameters has to possess some dense subgraphs.

\begin{lemma}
	\label{lem:tdensities}
    Let $\scrC = \{\cC_n\}$ be a family of $\dsl n,k(n), d(n) \dsr$ quantum LDPC codes with corresponding connectivity graphs $\cG = \{G_n\}_n$, and separation profiles $\{s_n\}_n$.
    The quantities $c_{\max}(n)$ and $r_0(n)$ are defined as in Definition \ref{def:cminmax}.
    Then for any connectivity graph $G_n$ of $\cC_n \in \cG$, the following two statements simultaneously hold:
	\begin{enumerate}
		\item $G_n$ is $\Omega(d(n))$-dense
		\item There exists $H_n \subset G_n$, such that $|H_n| \geq d(n)$, and $H_n$ is $|H_n|^{c_{\max}(n)}$-dense.
	\end{enumerate}
\end{lemma}
\begin{proof}
    First claim: from Lemma \ref{lem:generalizedbounds}, we have $d(n) \in O(s_n(n))$, or equivalently, $s_n(n) \in \Omega(d(n))$.
    This implies that $G_n$ is $\Omega(d(n))$-dense.

    Second claim: from the definition of $r_0$, there exists a subgraph $H_n \subset G_n$, with $|H_n| = r_0(n)$, such that $H_n$ is $r_0(n)^{c_{\max}(n)}$-dense.
    By choice of the optimization parameter, $r_0 \in [d(n),n]$, and therefore $|H_n| \geq d(n)$.
\end{proof}

\subsection{Proof of Theorem \ref{thm:main}}
\label{subsec:proofmain}

\begin{proof}
    \textbf{Claim 1:} $\eta_{\theta}$ induces $\Omega(d(n))$ edges of length $\Omega\left(\frac{d(n)}{\log(n)n^{(D-1)/D}}\right)$.
   
    From Lemma \ref{lem:tdensities}, we know that $G_n$ is $\Omega(d(n))$-dense.
    Therefore, by Lemma \ref{lem:denseedges}, we know that embedding $G_n$ requires at least $\Omega(d(n))$ edges of length $\Omega\left(\frac{d(n)}{\log(n)n^{(D-1)/D}}\right)$.
	
	\textbf{Claim 2:} $\eta_{\theta}$ induces $\Omega\left(\sqrt{\frac{k(n)}{n\log(n)^2}}d(n)\right)$ edges of length $\Omega\left(\sqrt{\frac{k(n)}{n\log(n)^2}}\frac{d(n)^{1/D}}{\log(n)}\right)$ for $n \geq n_0$.
    
    First, from Lemma \ref{lem:tdensities}, there exists $H_n \subset G_n$, such that $|H_n| \geq d(n)$, and $H_n$ is $|H_n|^{c_{\max}(n)}$-dense. 
    As we have shown before, a dense subgraph induces a large number of long edges: from Lemma \ref{lem:denseedges}, embedding $H_n$ requires $\Omega(|H_n|^{c_{\max}(n)})$ edges of length $\Omega\left(\frac{|H_n|^{1/D+c_{\max}(n) -1}}{\log(n)}\right)$.
    For these bounds to make sense, we now wish to lower bound $c_{\max}(n)$.
    
    This lower bound on $c_{\max}$ can be obtained from Lemma \ref{lem:generalizedbounds}. 
    It states that there exist $n_0, \beta$, such that for all $n \geq n_0$, $k(n)\leq \beta d(n)^{2(c_{\max}(n)-1)}n\log(n)^2$. 
    Then assuming that $n \geq n_0$, we have $c_{\max}(n) \geq \frac{\log(k(n)/\beta n \log(n)^2)}{2\log(d)}+1$.
    Further note that if $k(n)d(n)^{2/D} \geq \beta n\log(n)^2$, it can then be verified that $c_{\max}(n) \geq 1/D + c_{\max}(n) - 1 \geq 0$.
    Since $|H_n| \geq d(n)$, then this implies that $|H_n|^{c_{\max}(n)} \geq d(n)^{c_{\max}(n)} \geq d(n)^{\frac{\log(k(n)/\beta n\log(n)^2)}{2\log(d)}+1} = \sqrt{\frac{k(n)}{\beta n\log(n)^2}}d(n)$, and $|H_n|^{1/D + c_{\max}(n) -1 } \geq d(n)^{1/D + c_{\max}(n) -1 } \geq d(n)^{1/D + \frac{\log(k(n)/\beta n\log(n)^2)}{2\log(d)}} = \sqrt{\frac{k(n)}{\beta n\log(n)^2}}d(n)^{1/D}$.
    
    Also note that if $k(n)d(n)^{2/D} \leq \beta n\log(n)^2$, the bounds on the number and length of the edges become trivial, but still apply. 
	
	Since embedding $G_n$ implies embedding $H_n$, the desired result follows.

    \textbf{Claim 3:} If $k(n)d(n)^{2/D} \geq \beta n\log(n)^2/(1-\alpha)$, then $\eta_{\theta}$ induces at least $\Omega\left(\sqrt{\frac{(1-\alpha)k}{\beta n\log(n)^2}}^{1/\log_n(d)} \alpha k \right)$ edges of length $\Omega\left(\sqrt{\frac{(1-\alpha)k}{\beta n\log(n)^2}}\frac{d(n)^{1/D}}{\log(n)}\right)$ for $n \geq n_0$.

    In the proof of the previous claim, we have shown the existence of a subgraph $H_n$ such that $|H_n| \geq d(n)$ and $H_n$ is $|H_n|^{c_{\max}(n)}$-dense.
    However, is $H_n$ the \emph{only} dense subgraph whose existence we can prove? For a code to be good, is it sufficient to have only one dense subgraph? We will show that this is not the case.
    
    For the sake of clarity, we will write $H_{n,1}\equiv H_n$ and $G_{n,1} \equiv G_n$.
    We will be interested in what happens to $G_{n,1}$ when we remove $H_{n,1}$.

    To that end, we define $G_{n,2} = G_{n,1} \setminus H_{n,1}$, and we write $s_{n,2}(r)$ the separation profile of $G_{n,2}$.
    Similarly, we define $c_{\max,2}(n) =  \max_{r \in [d(n),n]} \log_r(s_{n,2}(r))$.
    
    The process of recursive separation of \cite{baspin2021connectivity} aims to find a tripartition $A,B,C$ of the qubits such that $A$ and $B$ are correctable.
    When applied to $G_{n,2}$, we can then find such $A,B,C$.
    We can then subsume  $H_{n,1}$ into $C$, which yield a tripartition $A,B,C \cup H_{n,1}$ of $G_{n,1}$ such that $A$ and $B$ are correctable.
    We then have $k \leq |C| + |H_{n,1}| \leq \beta  d^{2(c_{\max,2}(n)-1)}n\log(n)^2 + |H_{n,1}|$.
    
    If $H_{n,1}$ is small, and $c_{\max,2}(n)$ is too, then $k$ is actually more restricted than Lemma \ref{lem:generalizedbounds} might suggest.
    Another way to put it is, if we want $k \propto n$, then either $H_{n,1} \propto n$ ($H_{n,1}$ induces a large number of long edges) or $c_{\max,2}(n)$ has to be close to $1$, and therefore $H_{n,2}$ is also very dense. 
    Since $H_{n,1} \cap H_{n,2} = \emptyset$, then embedding $G_{n,1}$ requires embedding both $H_{n,1}$ and $H_{n,2}$ individually.  
    Equivalently, embedding $G_{n,1}$ requires embedding  $\Omega(|H_{n,1}|^{c_{\max,1}(n)} + |H_{n,2}|^{c_{\max,2}(n)})$, instead of merely $\Omega(|H_{n,1}|^{c_{\max,1}(n)})$ edges.

    We will extend these definitions to $G_{n,i}$, $H_{n,i}$ and $c_{\max,i}(n)$.
    At the $j$-th iteration the set $\{H_{n,i}\}_{i=1}^{i=j}$ induces $\Omega(\sum_{i=1}^{i=j} |H_{n,i}|^{c_{\max,i}(n)})$ edges, and $|H_{n,i}| \geq d$. 

    Using the above reasoning, we obtain at the $(j+1)$-th iteration
    \begin{align}
        \label{eq:klowbound}
        k \leq d^{2(c_{\max,j+1}(n)-1)}n + \sum_{i=1}^{i=j} |H_{n,i}|~.
    \end{align}
    Note that this peeling process might yield a stricter lower bound on $k$ only as long as $\sum_{i=1}^{i=j} |H_{n,i}|$ is of the order of $k$.
    We therefore consider the set of $\{H_{n,i}\}_{i=1}^{i=j+1}$ as the maximally large set such that $\sum_{i=1}^{i=j} |H_{n,i}| \leq \alpha k$, where $\alpha \in (0,1)$ is some fixed constant as in the theorem statement.

    As this set is maximally large, we cannot add any more elements; we have $\sum_{i=1}^{i=j+1} |H_{n,i}| > \alpha k$. 
    The number of edges to implement is then lower bounded by $\Omega(\sum_{i=1}^{i=j+1} |H_{n,i}|^{c_{\max,i}(n)} )$. 
    Further we have from eq.~\eqref{eq:klowbound} a lower bound on $c_{\max,i}(n)$
    \begin{align}
        c_{\max,i}(n)
        &\geq \frac{\log((k - \sum_{l=1}^{l=i-1} |H_{n,l}|)/\beta n\log(n)^2)}{2\log(d)} + 1 \nonumber\\
        &\geq \frac{\log((1-\alpha)k/\beta n\log(n)^2)}{2\log(d)} + 1 \label{eq:cmaxilowerbound}~.
    \end{align}
    If $k(n)d(n)^{2/D} \geq \beta n\log(n)^2/(1-\alpha)$ then eq.~\eqref{eq:cmaxilowerbound} implies $c_{\max,i}(n) \geq 1-1/D \geq 0$.
    We therefore have
    \begin{align*}
        \sum_{i=1}^{i=j+1} |H_{n,i}|^{c_{\max,i}(n)}
        &\geq \sum_{i=1}^{i=j+1} d^{\log_d(|H_{n,i}|)\frac{\log((1-\alpha)k/\beta n\log(n)^2)}{2\log(d)}}|H_{n,i}|\\
        &= \sum_{i=1}^{i=j+1} \sqrt{\frac{(1-\alpha)k}{\beta n\log(n)^2}}^{\log_n(|H_{n,i}|)/\log_n(d)}|H_{n,i}|~.
    \end{align*}
    Since $|H_{n,i}| \leq n$, then $\log_n(|H_{n,i}|) \leq 1$, and
    \begin{align*}
        \sum_{i=1}^{i=j+1} \sqrt{\frac{(1-\alpha)k}{\beta n\log(n)^2}}^{\log_n(|H_{n,i}|)/ \log_n(d)}|H_{n,i}|
        &\geq \sqrt{\frac{(1-\alpha)k}{\beta n\log(n)^2}}^{1/\log_n(d)} \sum_{i=1}^{i=j+1} |H_{n,i}|\\
        &\geq \sqrt{\frac{(1-\alpha)k}{\beta n\log(n)^2}}^{1/\log_n(d)} \alpha k~.
    \end{align*}

    The length of the edges is at least $\Omega(\min_{i \in [1,j+1]}|H_{n,i}|^{1/D + c_{\max,i}(n) - 1}/\log(n))$. By assumption, we have $c_{\max,i}(n) \geq 1-1/D$, or $1/D + c_{\max,i}(n) - 1 \geq 0$.
    Further, since $|H_{n,i}| \geq d$, we then have $|H_{n,i}|^{1/D + c_{\max,i}(n) - 1} \geq d^{1/D + c_{\max,i}(n) - 1} \geq d^{1/D}\sqrt{\frac{(1-\alpha)k}{\beta n\log(n)^2}}$, and the desired result follows.
\end{proof}

\section{Application of Main theorem to the stacked model}
\label{sec:stacked}
In this section, we return to Question 2 presented in the introduction: is it possible to implement a quantum LDPC code in $2$ or $3$ dimensions using \emph{mostly} local stabilizers?
We show that a particular model that has been proposed earlier, called the stacked architecture, provides strong evidence that the properties of such a  code will be limited.

We begin by describing the model in more detail.
Suppose we wished to design an error correcting code using a stacked layout in $2$-dimensions.
Consider the following proposal where qubits are laid out on a square grid of size $n = 2^{l_m} \times 2^{l_m}$ as shown in fig.~\ref{fig:stacked-pre}.
In total, there are $l_m$ layers in this stack, where the generators at level $l$ act within a ball of radius $r_l = 2^{l}/\sqrt{2}$.
At the very top, we have a highly nonlocal stabilizer associated with a ball of radius $r_{l_m} = 2^{l_m}/\sqrt{2}$.
To be clear, while the stabilizer in the top-most layer has a radius of $r_{l_m}$, it still only jointly measures some constant number of qubits, and each qubit is involved in a constant number of generators.
The radius merely constrains where these qubits are allowed to be located.
In the next layer we have $4$ stabilizers but these stabilizers are each only supported within a ball of radius $r_{l-1} = 2^{l-1}/\sqrt{2}$.
This proceeds until we hit the very last layer---there are $4^{l_m-l}$ such generators in layer $l$---until we hit layer $0$ which consists of stabilizers supported entirely within a ball of constant radius.
It follows that the majority of the stabilizers are in the last layer or in other words, the majority of stabilizers are local with $r = O(1)$ locality.
A natural question then is whether the nonlocal checks are numerous enough to allow for good codes.

We present two ways of obtaining bounds on the performance of such codes.
The first bound, presented in Section \ref{subsec:direct} is a direct application of Theorem \ref{thm:main} and is the tightest bound we could find.
The second bound, presented in Section \ref{subsec:metric} uses a metric that measures the average edge length, as well as the higher-order moments, to bound the properties of the code.
Although slightly weaker than the first bound, we present it as it may be applicable as a simple tool in other contexts.

\subsection{A direct bound}
\label{subsec:direct}
A corollary of our results is that the average length of the interactions in the implementation of a code limits code properties.
For example, a family of codes with linear distance requires $\Omega(n)$ edges of length $\Omegalog(n^{1/2})$.
If this system is sparse, then the average length is $\Omegalog(n^{1/2})$.
Conversely, if the average length of the interactions is not $\Omegalog(n^{1/2})$, then the system cannot implement a family of linear-distance codes.

Extending this idea, we can use a direct edge-counting argument together with Theorem \ref{thm:main} to bound the distance, and obtain a tradeoff between $k$ and $d$.

\begin{corollary}
	The stacked model satisfies $d = n^{2/3}\log(n)^{2/3}$, and $k^3d^4 = O(n^5\log(n)^4)$.
\end{corollary}
\begin{proof}
	
	We assume that each generator acts on at most $\delta_g$ qubits, and every qubit is contained in the support of at most $\delta_q$ generators. 
    Then each generator induces at most $\delta_0 \equiv \binom{\delta_g}{2}$ edges, and the degree of the connectivity graph is upper bounded by $(\delta_g-1)\delta_q$.
    
	Consider any set $M$ of edges in the connectivity graph of the stacked model.
    Then we are interested in the smallest of these edges, $e_t$, which lives in a layer $l_{t}$.
    Necessarily, all the other edges in $M$ live in the layers $[l_t,l_m]$, so $M$ is smaller than the number of edges living in the layers $[l_t,l_m]$.
    Further, each layer $l$ contains at most $4^{l_m-l}$ generators, and therefore induces at most  $\delta_0 4^{l_m-l}$ edges.
    We therefore have
	\[
	    |M| \leq \sum_{l_t}^{l_m} \delta_0 4^{l_m-l} \implies l_t \leq l_m -\log_4(|M|)+\alpha_0~,
    \]
    for some constant $\alpha_0$ depending on $\delta_0$. 
	
    Recall from Theorem \ref{thm:main}, Claim 1, that any embedding $\eta_{\theta}$ induces a set of $\Omega(d)$ edges of length $\Omega(\frac{d}{n^{1/2}\log(n)})$ in $2$-dimensions.
	We let $M$ correspond to this set, and therefore, $|M| = \Omega(d)$, or equivalently $\log_4(|M|) = \log_4(d) + \Omega(1)$. Also, since $e_t$ is in $M$, we have $|e_t|=\Omega\left(\frac{d}{n^{1/2}\log(n)}\right)$.
	Furthermore, by choice of $e_t$, it has length at most $2^{l_t}/\sqrt{2}$.
    We then have
	\[
	    2^{l_m -(\log_4(d)+ \Omega(1))+\alpha_0} = 2^{l_m -\log_4(|M|)+\alpha_0} \geq 2^{l_t} = \Omega(|e_t|) = \Omega\left( \frac{d}{n^{1/2}\log(n)}\right)~,
	\]
	where $|e_t|$ denotes the length of the edge $e_t$.
    Equivalently $d = O(n^{2/3}\log(n)^{2/3})$.

	Applying a similar analysis to the $k$-$d$ tradeoff from Theorem \ref{thm:main}, Claim 2, we obtain $k^3d^4 = O(n^5\log(n)^4)$.
\end{proof}

The distance bound immediately implies that this limited amount of nonlocality only yields a limited amount of leeway.
A $2$-dimensional local code, with this limited nonlocality, is constrained like a $3$-dimensional local code.
We do not know if this bound can be saturated, but it does not readily forbid the implementation of constant rate codes, with $d \propto\sqrt{n}$.

\subsection{A crude measure of nonlocality}
\label{subsec:metric}
In this section, we provide an alternate way to obtain (almost) the same bound as above.
We do so by proposing a crude measure of nonlocality which measures the average edge length of the connectivity graph and higher-order moments of the lengths of the edges.
We hope that this will be more broadly applicable as a quick-and-dirty tool when studying other implementations.

We begin with a metric which measures the length of an edge as given by a particular embedding.
\begin{definition}
    Let $\cC$ be an $\dsl n,k,d \dsr$ code with connectivity graph $G = (V,E)$.
    Let $\eta_{\theta}: V \to \bbR^D$ be an embedding.
    We define the nonlocality metric $\mu : E \rightarrow \mathbb{N}$ such that $\mu$ measures the length of an edge $e$ according to the embedding $\eta$.
    Equivalently,
	\begin{align*}
	    \forall e = (u,v) \in E, \mu(e) = \norm{\eta_{\theta}(u)-\eta_{\theta}(v)}~.
    \end{align*}
\end{definition}

We then define the $p$-th order of nonlocality as the $p$-th moment of the measure $\mu$.
\begin{definition}
	The $p$-th order of nonlocality, $p \in \mathbb{R}^+$, for a system with nonlocality metric $\mu$ is defined as
	\begin{align*}
        \Delta_p = \sum_{e \in E} \mu(e)^p.
    \end{align*}
\end{definition}

From Theorem \ref{thm:main}, we know that there have to be a minimum number of edges of a certain length.
This allows us to derive lower bounds on $\Delta_p$.
\begin{corollary}
    \label{cor:momentbounds}
	For a LDPC code $\cC$ with parameters $\dsl n,k,d \dsr$, we have
	\begin{enumerate}
		\item $\Delta_p  = \Omega\left(\frac{d^{p+1}}{n^{p(D-1)/D}\log(n)^p}  \right)$; and
		\item $\Delta_p  = \Omega\left(\left(\frac{k}{n}\right)^{(1+p)/2}\frac{d^{1+p/D}}{\log(n)^p}\right)$.
	\end{enumerate}
\end{corollary}
\begin{proof}
	First note that from Claim 1 in Theorem \ref{thm:main} there exist $\beta \in \Omega\left(\frac{d}{n^{(D-1)/D}\log(n)}\right)$, and $\gamma \in \Omega(d)$ such that there are at least $\gamma$ edges whose length is greater than $\beta$.
	We have
	\begin{align*}
		\Delta_p = \sum_{e \in E} \mu(e)^p 
		& \geq \sum_{e \in E : \mu(e) \geq \beta}\mu(e)^p \\
		&  \geq \beta^p \sum_{e \in E : \mu(e) \geq \beta} 1 \\
		& \geq \beta^p \gamma \\
		&=  \Omega\left(\frac{d^{p+1}}{n^{p(D-1)/D}\log(n)^p}  \right)
	\end{align*}

    The second claim follows similarly from Theorem \ref{thm:main}, claim 2.
\end{proof}

The next result focuses on the second moment.
It allows us to show the following distance and rate-distance tradeoffs.
\begin{corollary}
    \label{cor:stacktradeoff}
    Let $\cC$ be any $\dsl n,k,d \dsr$ code that is implement via the stacked architecture as described above.
    The embedding map $\eta_{\theta}$ is therefore implicit, where $\theta = 1$.
    The code $\cC$ satisfies $d = O(n^{2/3}\log(n))$, and $k^3 d^4 = O(n^{5}\log(n)^6)$.
\end{corollary}
\begin{proof}
	There are $4^{l_m-l}$ generators at level $l$.
    In the connectivity graph, each generator induces at most $\binom{\delta_g}{2} \equiv \delta_0$ interactions of length $2 r_l = 2 \cdot 2^l / \sqrt{2} = \sqrt{2}\cdot 2^l$.
	Then we have
    \begin{align*}
        \Delta_p \leq \sum_{l=0}^{l_m} \delta_0 4^{l_m-l} (2r_l)^p =  2^{p/2}\delta_0 \sum_{l=0}^{l_m}  4^{l_m-l}\cdot 2^{lp} = 2^{p/2}\delta_0 4^{l_m} \sum_{l=0}^{l_m} 2^{(p-2)l}~.
    \end{align*}
	
	This expression has then different closed forms, depending on the value of $p$.
    The optimum bound occurs at $p = 2$, then $\Delta_p = O(n\log(n))$.
    Using Corollary \ref{cor:momentbounds}, we have 
    \begin{align*}
        d = O(n^{2/3}\log(n)) \qquad k^3 d^4 = O(n^{5}\log(n)^6)~.
    \end{align*}
 	This concludes the proof.
 \end{proof}
This latter bound is weaker than that presented in Section \ref{subsec:direct} by polylogarithmic factors.
However, it is somewhat simpler in that it did not rely on the ordering of edge lengths.

\textbf{Discussion:}
\begin{enumerate}
\item We find that most known quantum LDPC codes do not violate these bounds.
Hypergraph product codes (with $k = \Theta(n)$ and $d = \Theta(\sqrt{n})$) \cite{tillich2014quantum}, codes based on high-dimensional expanders (with $k = \widetilde{\Theta}(\sqrt{n})$ and $d = \widetilde{\Theta}(\sqrt{n})$) \cite{evra2020decodable,kaufman2020quantum}, 2D hyperbolic codes (with $k = \Theta(n)$ and $d = O(\log(n))$)\cite{freedman2002z2,breuckmann2016constructions}, 4D hyperbolic codes (with $k= \Theta(n)$ and $d= O(n^{\epsilon})$ for $\epsilon < 0.3$) \cite{londe2017golden,hastings2013decoding,guth2014quantum}, fiber bundle codes (with $k = \widetilde{\Theta}(n^{3/5})$ and $d = \widetilde{\Theta}(n^{3/5})$)\cite{hastings2020fiber} or balanced-product codes (with $k = \widetilde{\Theta}(n^{4/5})$ and $d = \widetilde{\Theta}(n^{3/5}$))\cite{breuckmann2020balanced} do not violate either the distance bound or the $k$-$d$ tradeoff.
Indeed all of these codes have a distance that scales as $o(n^{2/3})$.
Of these codes, the only constant-rate codes are hypergraph product codes and hyperbolic codes.

It is still not clear whether these codes \emph{can} be implemented via a stacked architecture, but our techniques do not rule out this possibility.
It would be interesting to find an explicit layout of a hypergraph product code, the best constant-rate codes, in two dimensions which can be implemented using such a model.
\item On the other hand, the Panteleev-Kalachev codes \cite{panteleev2020quantum} achieve distance $\Theta(n/\log(n))$; these codes clearly violate the distance bound.
In general, their codes achieve $d = \Theta(n^{1-\alpha/2}/\log(n))$ and $k = \Theta(n^{\alpha}\log(n))$.
These are not ruled out by the above bounds when $\alpha \in (2/3,1)$.
We also note that although codes of distance $d = \Theta(n)$ are not known to exist, these cannot be implemented using a stacked architecture.
\end{enumerate}

\section{Conclusions}
We considered the question of how much nonlocality is needed to implement quantum LDPC codes.
In our results, this question is addressed by lower bounding the number of long-range connections between qubits, and their length.
In particular, in 2D we show that a quantum LDPC code with distance $d \propto n^{1/2 + \epsilon}$ requires $\Omega(n^{1/2 + \epsilon})$ interactions of length $\Omegalog(n^{\epsilon})$.
We also focus on constant-rate quantum LDPC codes, as the cost of encoding a logical qubit in such a code remains fixed.
For such a code to exhibit a distance $d \propto n^\alpha$, we find that one requires $\Omegalog(n)$ interactions of length $\Omegalog(n^{\alpha/2})$.
We then considered a stacked architecture, a model considered to implement quantum LDPC codes.
In this model, although most stabilizers are local, a few are capable of longe-range connections.
We showed that the distance of this architecture is bounded.
Furthermore, it too witnesses a sharp tradeoff between $k$ and $d$.
We hope these tools can be used to understand the difficulty of implementing efficient codes, as well as the limitations of particular architectures.

\emph{Acknowledgements}---
We would like to thank Guillaume Duclos-Cianci for facilitating this collaboration.
We thank Anthony Leverrier for his comments on a draft of this manuscript.
AK is supported by the Bloch postdoctoral fellowship at Stanford University and NSF grant CCF-1844628.
AK thanks Emily Davis, Dripto Debroy and Sam Roberts for pointing him to references on implementations of long-range connectivity.

\bibliographystyle{unsrtabbrev}
\bibliography{references}

\begin{thebibliography}{10}

\bibitem{aharonov1997fault}
D.~Aharonov and M.~Ben-Or.
\newblock Fault-tolerant quantum computation with constant error.
\newblock In {\em Proceedings of the twenty-ninth annual ACM symposium on
  Theory of computing}, pages 176--188. ACM, 1997.

\bibitem{kitaev1997quantum}
A.~Y. Kitaev.
\newblock Quantum computations: algorithms and error correction.
\newblock {\em Russian Mathematical Surveys}, 52(6):1191--1249, 1997.

\bibitem{knill1998resilient}
E.~Knill, R.~Laflamme, and W.~H. Zurek.
\newblock Resilient quantum computation: error models and thresholds.
\newblock In {\em Proceedings of the Royal Society of London A: Mathematical,
  Physical and Engineering Sciences}, volume 454, pages 365--384. The Royal
  Society, 1998.

\bibitem{aliferis2005quantum}
P.~Aliferis, D.~Gottesman, and J.~Preskill.
\newblock Quantum accuracy threshold for concatenated distance-3 codes.
\newblock {\em arXiv preprint quant-ph/0504218}, 2005.

\bibitem{ogorman2017quantum}
J.~O'Gorman and E.~T. Campbell.
\newblock Quantum computation with realistic magic-state factories.
\newblock {\em Physical Review A}, 95(3):032338, 2017.

\bibitem{sanders2020compilation}
Y.~R. Sanders, D.~W. Berry, P.~C. Costa, L.~W. Tessler, N.~Wiebe, C.~Gidney,
  H.~Neven, and R.~Babbush.
\newblock Compilation of fault-tolerant quantum heuristics for combinatorial
  optimization.
\newblock {\em PRX Quantum}, 1(2):020312, 2020.

\bibitem{Note1}
We often refer to $k$ as the number of \protect \emph {logical} qubits, or as
  the \protect \emph {dimension} of the code.

\bibitem{bravyi2009no}
S.~Bravyi and B.~Terhal.
\newblock A no-go theorem for a two-dimensional self-correcting quantum memory
  based on stabilizer codes.
\newblock {\em New Journal of Physics}, 11(4):043029, 2009.

\bibitem{bravyi2010tradeoffs}
S.~Bravyi, D.~Poulin, and B.~Terhal.
\newblock Tradeoffs for reliable quantum information storage in 2{D} systems.
\newblock {\em Physical Review Letters}, 104(5):050503, 2010.

\bibitem{bergeron2020silicon}
L.~Bergeron, C.~Chartrand, A.~Kurkjian, K.~Morse, H.~Riemann, N.~Abrosimov,
  P.~Becker, H.-J. Pohl, M.~Thewalt, and S.~Simmons.
\newblock Silicon-integrated telecommunications photon-spin interface.
\newblock {\em PRX Quantum}, 1(2):020301, 2020.

\bibitem{bombin2021interleaving}
H.~Bombin, I.~H. Kim, D.~Litinski, N.~Nickerson, M.~Pant, F.~Pastawski,
  S.~Roberts, and T.~Rudolph.
\newblock Interleaving: Modular architectures for fault-tolerant photonic
  quantum computing.
\newblock {\em arXiv preprint arXiv:2103.08612}, 2021.

\bibitem{periwal2021programmable}
A.~Periwal, E.~S. Cooper, P.~Kunkel, J.~F. Wienand, E.~J. Davis, and
  M.~Schleier-Smith.
\newblock Programmable interactions and emergent geometry in an atomic array.
\newblock {\em arXiv preprint arXiv:2106.04070}, 2021.

\bibitem{monroe2014large}
C.~Monroe, R.~Raussendorf, A.~Ruthven, K.~Brown, P.~Maunz, L.-M. Duan, and
  J.~Kim.
\newblock Large-scale modular quantum-computer architecture with atomic memory
  and photonic interconnects.
\newblock {\em Physical Review A}, 89(2):022317, 2014.

\bibitem{linke2017experimental}
N.~M. Linke, D.~Maslov, M.~Roetteler, S.~Debnath, C.~Figgatt, K.~A. Landsman,
  K.~Wright, and C.~Monroe.
\newblock Experimental comparison of two quantum computing architectures.
\newblock {\em Proceedings of the National Academy of Sciences},
  114(13):3305--3310, 2017.

\bibitem{murali2020architecting}
P.~Murali, D.~M. Debroy, K.~R. Brown, and M.~Martonosi.
\newblock Architecting noisy intermediate-scale trapped ion quantum computers.
\newblock In {\em 2020 ACM/IEEE 47th Annual International Symposium on Computer
  Architecture (ISCA)}, pages 529--542. IEEE, 2020.

\bibitem{bombin2006topological}
H.~Bombin and M.~A. Martin-Delgado.
\newblock Topological quantum distillation.
\newblock {\em Physical Review Letters}, 97(18):180501, 2006.

\bibitem{bravyi1998quantum}
S.~Bravyi and A.~Y. Kitaev.
\newblock Quantum codes on a lattice with boundary.
\newblock {\em arXiv preprint quant-ph/9811052}, 1998.

\bibitem{kubica2015universal}
A.~Kubica and M.~E. Beverland.
\newblock Universal transversal gates with color codes: A simplified approach.
\newblock {\em Physical Review A}, 91(3):032330, 2015.

\bibitem{gottesman2014fault}
D.~Gottesman.
\newblock Fault-tolerant quantum computation with constant overhead.
\newblock {\em Quantum Information \& Computation}, 14(15-16):1338--1372, 2014.

\bibitem{kovalev2013fault}
A.~A. Kovalev and L.~P. Pryadko.
\newblock Fault tolerance of quantum low-density parity check codes with
  sublinear distance scaling.
\newblock {\em Physical Review A}, 87(2):020304, 2013.

\bibitem{fawzi2018constant}
O.~Fawzi, A.~Grospellier, and A.~Leverrier.
\newblock Constant overhead quantum fault-tolerance with quantum expander
  codes.
\newblock In {\em 2018 IEEE 59th Annual Symposium on Foundations of Computer
  Science (FOCS)}, pages 743--754. IEEE, 2018.

\bibitem{baspin2021connectivity}
N.~Baspin and A.~Krishna.
\newblock Connectivity constrains quantum codes.
\newblock {\em arXiv preprint arXiv:2106.00765}, 2021.

\bibitem{Note2}
We are not aware of the origin of this model. We heard about it through David
  Poulin; something similar was also mentioned by Daniel Gottesman \cite
  {gottesman2021talk}.

\bibitem{tillich2014quantum}
J.-P. Tillich and G.~Z{\'e}mor.
\newblock Quantum {LDPC} codes with positive rate and minimum distance
  proportional to the square root of the blocklength.
\newblock {\em IEEE Transactions on Information Theory}, 60(2):1193--1202,
  2014.

\bibitem{Note3}
Our understanding of how to perform fault-tolerant quantum computation using
  only a single block of an LDPC code is limited; see for example \cite
  {krishna2021fault}.

\bibitem{matouvsek2013lecture}
J.~Matou{\v{s}}ek.
\newblock Lecture notes on metric embeddings.
\newblock Technical report, Technical report, ETH Z{\"u}rich, 2013.

\bibitem{hogben2013handbook}
L.~Hogben.
\newblock {\em Handbook of Linear Algebra, Second Edition}.
\newblock Discrete Mathematics and Its Applications. Chapman and Hall/CRC, 2
  edition, 2013.

\bibitem{alon2016probabilistic}
J.~H.~S. Noga~Alon.
\newblock {\em The Probabilistic Method, 4th Edition}.
\newblock Wiley Series in Discrete Mathematics and Optimization. John Wiley \&
  Sons, 4ed. edition, 2016.

\bibitem{bottcher2010bandwidth}
J.~B{\"o}ttcher, K.~P. Pruessmann, A.~Taraz, and A.~W{\"u}rfl.
\newblock Bandwidth, expansion, treewidth, separators and universality for
  bounded-degree graphs.
\newblock {\em European Journal of Combinatorics}, 31(5):1217--1227, 2010.

\bibitem{evra2020decodable}
S.~Evra, T.~Kaufman, and G.~Z{\'e}mor.
\newblock Decodable quantum {LDPC} codes beyond the square root distance
  barrier using high dimensional expanders.
\newblock In {\em 2020 IEEE 61st Annual Symposium on Foundations of Computer
  Science (FOCS)}, pages 218--227. IEEE, 2020.

\bibitem{kaufman2020quantum}
T.~Kaufman and R.~J. Tessler.
\newblock New cosystolic expanders from tensors imply explicit quantum {LDPC}
  codes with {$\Omega(\sqrt{n}\log^{k}n)$} distance.
\newblock page 1317–1329, 2021.

\bibitem{freedman2002z2}
M.~H. Freedman, D.~A. Meyer, and F.~Luo.
\newblock Z2-systolic freedom and quantum codes.
\newblock {\em Mathematics of quantum computation, Chapman \& Hall/CRC}, pages
  287--320, 2002.

\bibitem{breuckmann2016constructions}
N.~P. Breuckmann and B.~M. Terhal.
\newblock Constructions and noise threshold of hyperbolic surface codes.
\newblock {\em IEEE transactions on Information Theory}, 62(6):3731--3744,
  2016.

\bibitem{londe2017golden}
V.~Londe and A.~Leverrier.
\newblock Golden codes: quantum {LDPC} codes built from regular tessellations
  of hyperbolic 4-manifolds.
\newblock {\em arXiv preprint arXiv:1712.08578}, 2017.

\bibitem{hastings2013decoding}
M.~B. Hastings.
\newblock Decoding in hyperbolic spaces: Quantum {LDPC} codes with linear rate
  and efficient error correction.
\newblock {\em Quantum Info. Comput.}, 14(13–14):1187–1202, Oct. 2014.

\bibitem{guth2014quantum}
L.~Guth and A.~Lubotzky.
\newblock Quantum error correcting codes and 4-dimensional arithmetic
  hyperbolic manifolds.
\newblock {\em Journal of Mathematical Physics}, 55(8):082202, 2014.

\bibitem{hastings2020fiber}
M.~B. Hastings, J.~Haah, and R.~O'Donnell.
\newblock Fiber bundle codes: Breaking the $n^{1/2}$poly$\log(n)$ barrier for
  quantum {LDPC} codes.
\newblock page 1276–1288, 2021.

\bibitem{breuckmann2020balanced}
N.~P. Breuckmann and J.~N. Eberhardt.
\newblock Balanced product quantum codes.
\newblock {\em IEEE Transactions on Information Theory}, pages 1--1, 2021.

\bibitem{panteleev2020quantum}
P.~Panteleev and G.~Kalachev.
\newblock Quantum {LDPC} codes with almost linear minimum distance.
\newblock {\em arXiv preprint arXiv:2012.04068}, 2020.

\bibitem{gottesman2021talk}
D.~Gottesman.
\newblock Fault tolerance with {LDPC} codes.
\newblock Talk available at \url{https://www.youtube.com/watch?v=PD4h6ZIV2gY}.

\bibitem{krishna2021fault}
A.~Krishna and D.~Poulin.
\newblock Fault-tolerant gates on hypergraph product codes.
\newblock {\em Phys. Rev. X}, 11:011023, Feb 2021.

\end{thebibliography}

\end{document}